%% file: main.tex
\documentclass[a4paper,UKenglish,cleveref, autoref, thm-restate]{lipics-v2021}

\usepackage{subfiles}
\newcommand*{\bibsubfile}[2]{
\bibliographystyle{#1}
\bibliography{#2}
}

\usepackage{standalone}
\usepackage{amsthm}
\usepackage{mathtools}
\usepackage{bbm}
\usepackage[ruled,vlined,linesnumbered,commentsnumbered]{algorithm2e}
\usepackage{tikz}
\usetikzlibrary{calc,arrows.meta,decorations,decorations.pathmorphing,patterns}

\allowdisplaybreaks
\DeclareGraphicsExtensions{.png,.pdf}
\graphicspath{{img/}{../img/}}
\usepackage{mdwlist}

\usepackage{thmtools}
\usepackage{thm-restate}
\usepackage{diagbox}
\usepackage{filecontents}
\usepackage[colorinlistoftodos]{todonotes}
\usepackage[font={footnotesize},labelsep=colon]{caption}

\usepackage{flushend}
\usepackage{multirow}
\usepackage{cleveref}
\usepackage{makecell}

\include{math_notations}

\RequirePackage{xparse}
\RequirePackage{mleftright}

\crefformat{footnote}{#2\footnotemark[#1]#3}

\title{Correlated Stochastic Knapsack with a Submodular Objective} 

\titlerunning{Correlated Stochastic Knapsack with a Submodular Objective} 

\author{Sheng Yang}{Northwestern University}{styang@fastmail.com}{https://orcid.org/0000-0002-5884-4893}{}
\author{Samir Khuller}{Northwestern University}{samir.khuller@northwestern.edu}{}{}
\author{Sunav Choudhary}{Adobe Research}{schoudha@adobe.com}{}{}
\author{Subrata Mitra}{Adobe Research}{subrata.mitra@adobe.com}{}{}
\author{Kanak Mahadik}{Adobe Research}{mahadik@adobe.com}{}{}

\authorrunning{S. Yang, S. Khuller, S. Choudhary, S. Mitra, K. Mahadik} 

\Copyright{Sheng Yang, Samir Khuller, Sunav Choudhary, Subrata Mitra, Kanak Mahadik} 


\category{} 

\relatedversion{} 

\funding{S. Khuller and S. Yang were funded by an Amazon Research Award, and Adobe Research.}

\nolinenumbers 

\hideLIPIcs  

\EventEditors{John Q. Open and Joan R. Access}
\EventNoEds{2}
\EventLongTitle{42nd Conference on Very Important Topics (CVIT 2016)}
\EventShortTitle{CVIT 2016}
\EventAcronym{CVIT}
\EventYear{2016}
\EventDate{December 24--27, 2016}
\EventLocation{Little Whinging, United Kingdom}
\EventLogo{}
\SeriesVolume{42}
\ArticleNo{23}

\begin{document}

\maketitle

\input{abstract}

\pagebreak
\setcounter{page}{1}

\renewcommand{\bibsubfile}[2]{}

\input{intro_neo}
\input{related}
\input{formulation}
\input{continuous}
\input{rounding}
\input{conclusion}
\section*{Acknowledgement}
We would like to thank the reviewers for comments on a previous version of this draft, who found a flaw in the algorithm and its analysis. The old algorithm imposes the partition matroid \emph{while} doing the first rounding step, leading to dependency between items. This seemingly convenient step actually breaks the correctness of contention resolution scheme,  which is built on FKG inequality and intrinsically needs an independent rounding step. We fixed the issue by replacing it with a true independent rounding step, and fix the solution to fit the partition matroid later on. While this breaks the symmetry between items, the gap of 2 turns out to be large enough to fix everything. Check the use of union bound in Case 1 for the proof of \Cref{lemma:case} for details.

\bibliography{bibfile}
\appendix
\input{proofs}
\input{scheduling}

\end{document}

%% file: math_notations.tex
\newcommand{\NN}{\mathcal{N}}
\newcommand{\PP}{\mathcal{P}}

\newcommand{\cS}{\mathcal{S}}
\newcommand{\cP}{\mathcal{P}}
\newcommand{\cF}{\mathcal{F}}
\newcommand{\cI}{\mathcal{I}}

\newcommand{\fS}{\boldsymbol{\mathcal{S}}}
\newcommand{\fA}{\boldsymbol{\mathcal{A}}}
\newcommand{\fP}{\boldsymbol{\mathcal{P}}}

\newcommand{\R}{\mathbb{R}}
\newcommand{\E}{\mathbb{E}}

\newcommand{\size}{\textsf{size}}

\newcommand{\ExpP}{\textsf{ExpP}}
\newcommand{\PolyP}{\textsf{PolyP}}

\newcommand{\OPT}{\textsf{OPT}}

\RequirePackage{xparse}
\RequirePackage{mleftright}
\NewDocumentCommand\NewPrefixPostfixDelimiter{mmmm}{%
	\NewDocumentCommand#1{om}{%
		\IfNoValueTF{##1}
		{#2#3{##2}#4}
		{#2#3[##1]{##2}#4}%
	}%
}
\NewDocumentCommand\NewPairedDelimiter{mmm}{%
	\NewDocumentCommand#1{som}{%
		\IfNoValueTF{##2}{%
			\IfBooleanTF{##1}
			{#2##3#3}
			{\mleft#2##3\mright#3}%
		}
		{\mathopen{##2#2}##3\mathclose{##2#3}}%
	}%
}
\NewPairedDelimiter{\bb}{(}{)}
\NewPairedDelimiter{\bB}{(}{]}
\NewPairedDelimiter{\Bb}{[}{)}
\NewPrefixPostfixDelimiter{\BigOh}{O}{\bb}{}

%% file: abstract.tex
\begin{abstract}
  \label{sec:abstract}
  We study the correlated stochastic knapsack problem of a submodular target function, with optional additional constraints.
  We utilize the multilinear extension of submodular function, and bundle it with an adaptation of the relaxed linear constraints from Ma [Mathematics of Operations Research, Volume 43(3), 2018] on correlated stochastic knapsack problem. The relaxation is then solved by the stochastic continuous greedy algorithm, and rounded by a novel method to fit the contention resolution scheme (Feldman et al. [FOCS 2011]). We obtain a pseudo-polynomial time \((1 - 1/\sqrt{e})/2 \simeq 0.1967\) approximation algorithm with or without those additional constraints, eliminating the need of a key assumption and improving on the \((1 - 1/\sqrt[4]{e})/2 \simeq 0.1106\) approximation by Fukunaga et al. [AAAI 2019].

\end{abstract}


%% file: intro_neo.tex
\section{Introduction}
\label{sec:intro}

The knapsack problem is one of the most celebrated frameworks to model profit maximization with limited resources. Though well understood in its basic form, many new variants were proposed to model and target more complicated problems.
One line of variants take randomness into consideration.
Such randomness may appear on item sizes only, on item profits only, or on both in a correlated fashion. A significant body of work~\cite{bhalgat2011improved,dean2005adaptivity,gupta2011approximation,li2013stochastic,ma2018improvements} connects knapsack problem with the field of stochastic optimization, greatly broadening the spectrum of knapsack problems while introducing various challenges for theoretical analysis.
Another line of variants model diminishing returns in the profit, leading to the field of submodular optimization~\cite{nemhauser1978analysis,fisher1978analysis,calinescu2011maximizing,chekuri2009dependent,feldman2013maximization,feldman2011unified},  which enjoys tremendous popularity both in theory and in practice.
The two lines of work are connected together into stochastic submodular optimization, another fruitful field~\cite{asadpour2008stochastic,golovin2011adaptive,chen2013near,fujii2016budgeted,gabillon2013adaptive,gabillon2014large,gotovos2015nonmonotone,yu2016linear}.
In this work, we follow this line, and consider a correlated stochastic knapsack problem with a submodular target function. We arrived at this problem when modeling the spot scheduling problem in Yang et al.~\cite{yang2021scheduling} (see details in \Cref{sec:spot_problem}). A slight variant of the final problem was first proposed in Fukunaga et al.~\cite{fukunaga2019stochastic}, trying to model ``performance-dependent costs of items'' in stochastic submodular optimization. This problem turns out to be a very powerful framework that applies to several other real world applications, like recommendation systems~\cite{yue2011linear,ahmed2012fair}, and batch-mode active learning~\cite{hoi2006batch}.

\subsection{Formal Problem Statement}
\label{sec:spot_final_problem}
There are \(n\) items, each takes a random \(\size_{i}\in \mathbb{N}\) with probability \(p_{i}(\size_{i})\), and gets a reward that \emph{corresponds} to its size.  In other words, for each item \(i\), there is a reward function \(R_{i}: \mathbb{N}\rightarrow [M] \), such that \(r_{i} = R_{i}(\size_{i})\). (For simplicity, we define \([n]\) to be the set \(\{0, 1, 2, \dots, n\}\), and \(M\) a positive integer that upper bounds the maximum reward.) We assume each \(R_{i}\) to be non-decreasing, i.e., the larger an item, the more reward it deserves. We are given a budget \(B \in \mathbb{N}\) for the total size of items, and wish to extract as much profit as possible. The total profit is a lattice-submodular function\footnote{See definition of partition matroid, submodular and lattice-submodular in \Cref{sec:formulation}.} \(f: [M]^{n} \rightarrow \mathbb{R}^{ + }\) on the rewards of included items, and we wish to maximize its expectation\footnote{%
Let \(S \subseteq [n]\), we sample a vector \(q\in [M]^{n}\) as follows. Each component \(q(i)\) is sampled independently.  For \(i\in S\), \(\Pr[r_{i} = R_{i}(s)] = p_{i}(s)\); for \(i \notin S\), \(r_{i} = 0\) with probability 1. Denote this distribution as \(q_{S}\). Then the objective is to select a (random) set \(S\subseteq \mathcal{I}\) of items that maximizes \(\E_{\theta \sim q_{S}}[f(\theta)]\) subject to \(\sum_{i\in S} \size_{i} \leq B\).}.

Items are put in the knapsack one by one. As soon as an item is put in the knapsack, its reward and size are revealed. We halt when the knapsack overflows (not collecting the last item's reward), and proceed to add another item otherwise. We consider adaptive policies, i.e., we can choose an item to include, observe its realized size, and make further decisions based on the realized size. At first, only the reward function and the size distribution of items are known. When the policy includes an item \(i\), its \(\size_{i}\) is realized, and so is its reward \(r_{i} = R_{i}(\size_{i})\). In this work, we only consider adaptive policies without cancellation, i.e., the policy can make its decision based on all the realizations it has seen so far, and the inclusion of an item is irrevocable.

For a vector \(q\in [M]^{n}\), let \(\Pr_{\gamma}[q]\) denote the probability that we get outcome \(q\) when running policy \(\gamma\). Note this probability is with respect to the randomness both in the state of items and in the policy \(\gamma\). Let \(f_{\textsf{avg}}(\gamma)\) denote \(\sum_{q\in [M]^{n}} \Pr_{\gamma}[q] f(q)\), i.e., the average objective value obtained by \(\gamma\). Our aim is to find a policy \(\gamma\) that maximizes \(f_{\textsf{avg}}(\gamma)\). We say \(\gamma\) is an \(\alpha\)-approximation policy if \(f_{\textsf{avg}} (\gamma) \geq \alpha f_{\text{avg}}(\gamma^{*})\) for any policy \(\gamma^{ * }\).

In addition to all the above, we further require that the chosen set of items \(S\) be an independent set of a partition matroid\footnote{\label{fn:def}See definition of partition matroid, submodular and lattice-submodular in \Cref{sec:formulation}.} \(\mathcal{I} = \{\cI_{k}\}_{k\in [K]}\).
This is without loss of generality as we can put each item in a separate partition, and every subset of items is valid.
The additional constraint allows us to impose conflicts between items, which is needed for the modeling in Yang et al.~\cite{yang2021scheduling}.
More importantly, it is also crucial if we are to allow the attempt to include an item that could possibly overflow the knapsack, a case unsolved and left as open problem in Fukunaga et al.~\cite{fukunaga2019stochastic} (see details in \Cref{sec:spot_eliminate}). This partition matroid is also used to ensure the correctness of our approach based on a time-indexed LP.

\subsection{Our Contributions}
We present a pseudo-polynomial time algorithm for the correlated stochastic knapsack problem with a submodular target function. It computes an adaptive policy for this problem which is guaranteed to achieve \((1 - 1/\sqrt{e})/2 \simeq 0.1967\) of the optimal solution on expectation. It improves on the \((1 - 1/\sqrt[4]{e})/2 \simeq 0.1106\) approximation algorithm from Fukunaga et al.~\cite{fukunaga2019stochastic}. Furthermore, we eliminate one key assumption in Fukunaga et al.~\cite{fukunaga2019stochastic} which does not allow the inclusion of any item which could possible overflow the budget.

\subsection{Eliminating An Assumption in Previous Work}
\label{sec:spot_eliminate}

In Fukunaga et al.~\cite{fukunaga2019stochastic}, the authors considered a slightly different problem.
They made two assumptions, and we managed to eliminate one of them. The first assumption states that larger size means larger reward for every particular job.  This assumption is reasonable for general problems and remains crucial in our analysis.
The second assumption states that we will never select an item which could overflow the budget, given the realization of selected items\footnote{For example, suppose we are left with a remaining budget of 20 at some time, and all items have a 0.001 probability of size 21. What this assumption suggests is that none of the items are allowed to be selected.}. However, for many cases, selecting such an item is a desirable choice since additional value is obtained with high probability. If we are unlucky and the size goes beyond the remaining budget, we either receive a partial value, or do not get any value at all.

\subsection{Our Techniques}
\label{sec:spot_technique}

If the target function is linearly additive, this problem becomes the correlated stochastic knapsack problem.
For this problem, Gupta et al.~\cite{gupta2011approximation} gave an \(1/8\) approximation algorithm for adaptive policies based on LP relaxation.
The approximation ratio was improved to \(1/(2 + \epsilon)\) by Ma~\cite{ma2018improvements}, via a different time indexed LP formulation and a more sophisticated rounding scheme.
Fukunaga et al.~\cite{fukunaga2019stochastic} extends the \(1/8\) approximation algorithm, and achieve a \((1 - 1/\sqrt[4]{e})/2\) approximation for a case with submodular target function.
This is achieved via a combination of the stochastic continuous greedy algorithm~\cite{asadpour2016maximizing} (for getting a fractional solution), and the contention resolution scheme~\cite{feldman2011unified} (for rounding).
A natural idea for improvement is to take ingredients from the \(1 / (2 + \epsilon)\) algorithm by Ma~\cite{ma2018improvements}.
While the LP can be easily adapted, its rounding exhibits complicated dependencies that can be hard to analyze.
We also have no luck with a direct application of the \emph{contention resolution scheme}\cite{chekuri2014submodular,feldman2013maximization,feldman2011unified}, a powerful technique in submodular optimization.
The difficulties come in two folds. First, the original scheme is based on FKG inequality, which requires  an \emph{independently} rounded solution (possibly invalid) to start with. Any attempt to enforce our partition matroid at this step will break the whole scheme. This invalid solution is later fixed by ignoring some items from the rounded solution, and we need a way to impose the additional partition matroid constraint. Second, the ignoring step needs a critical ``monotone'' property.
At a high level, the property says that the more items you choose, the lower the probability every other items will be selected (See \Cref{sec:CRS} for a rigours definition).
While this may seem trivially true for any reasonable algorithm, it is not. In particular, it does not hold for Ma's algorithm~\cite{ma2018improvements}, due to its complicated dependencies.
The first difficulty is not hard: if we happen to pick two items from the same partition, we just throw the later one out. Unfortunately, this makes the second obstacle even harder.
The second obstacle is overcome by designing a brand-new rounding scheme which allows the direct analysis on the correlated probability of events.
In order to achieve the aforementioned monotone property, we insert phantom items to block some ``time slots'' even when no item is there to conflict with. Such phantom items may be of independent interest for other applications of the contention resolution scheme.
This alternate way of achieving monotonicity simultaneously free our analysis from one assumption mentioned in \Cref{sec:spot_eliminate}, which was needed in Fukunaga et al.~\cite{fukunaga2019stochastic} in their proof of the monotone property.
A factor of \((1 - 1/\sqrt{e})\) is lost for the continuous optimization part, and another factor of \(2\) is lost for rounding, leading to our \((1 - 1/\sqrt{e})/2 \simeq 0.1967\) approximation algorithm.


%% file: related.tex
\section{Other Related Works}
\label{sec:related_works}

\textbf{Stochastic Knapsack Problem}
\label{sec:related_stochastic}
The stochastic version of the knapsack problem has long been studied.
Kleinberg et al.~\cite{kleinberg2000allocating}, and Goel and Indyk~\cite{goel1999stochastic} consider the stochastic version to maximize profit that will overflow the budget with probability at most $p$.
However, they assume deterministic profits and special size distributions.
Dean et al.~\cite{dean2008approximating} relax the limit on size and allow arbitrary distributions for item sizes.
They investigate the gap between non-adaptive policies (the order of items to insert is fixed) and adaptive policies (allowed to make dynamic decision based on the realized size of items) and give a polynomial-time non-adaptive algorithm that approximates the optimal adaptive policy within a factor of \(1/4\) in expectation.
They also give an adaptive policy that approximates within a factor of $1/(3 + \epsilon)$ for any constant $\epsilon > 0$.
Bhalgat et al.~\cite{bhalgat2011improved} improves on this and give a bi-criteria \((1 - \epsilon)\) algorithm by relaxing the budget by \((1 + \epsilon)\).
Dean et al.~\cite{dean2005adaptivity} show that if correlation between size and reward is allowed, the problem would be PSPACE-hard.
Gupta et al.~\cite{gupta2011approximation} considered the case where the size and reward of an item can be arbitrarily correlated, and give an \(1/8\) approximation.
Li and Yuan~\cite{li2013stochastic} improved on this and get a $1/(2 + \epsilon)$ approximation with correlations and cancellation when $\epsilon$ fraction of extra space is allowed.
This was further improved by Ma~\cite{ma2018improvements}, who gets the same approximation ratio but without the budget augmentation requirement.

\textbf{Submodular Maximization}
\label{sec:related_submodular}
Nemhauser et al.~\cite{nemhauser1978analysis} studied the problem of maximizing a monotone submodular function subject to a cardinality constraint and gave the standard greedy \((1 - 1/e)\)-approximation algorithm.
For the case with a matroid constraint, Fisher et al.~\cite{fisher1978analysis} showed that the standard greedy algorithm gives a \(1/2\)-approximation. This was improved to \((1 - 1/e)\) by Calinescu et al.~\cite{calinescu2011maximizing}, via the \emph{continuous greedy algorithm}, which was originally developed by Calinescu et al.~\cite{calinescu2007maximizing} for the submodular welfare problem. In this algorithm, the target function is relaxed via an exponential multilinear-extension. Though exponential, this version can be approximately solved to arbitrary precision in polynomial time. The fractional solution is then rounded via pipage rounding~\cite{calinescu2007maximizing,calinescu2011maximizing,vondrak2009symmetry} or other rounding schemes~\cite{chekuri2009dependent}. In order to generalize the problem for other constraints and non-monotone submodular functions, a general rounding framework \emph{contention resolution scheme} was proposed~\cite{calinescu2011maximizing,feldman2013maximization,feldman2011unified}.
In this framework, the rounding step happens in two phases, an independent rounding phase followed by a pruning phase, where the second phase ensures an upper bound on the probability that an element is pruned.
One line of stochastic submodular optimization~\cite{asadpour2008stochastic} assumes items have stochastic states, and would like to maximize a monotone submodular function on the stochastic states, under constraints on the set of chosen items.
In other words, the constraints only depend on the selection of items, but not on the stochastic states of them.
This is a generalization of the stochastic knapsack problem where the size of items are deterministic. Various settings of this problem are investigated by a series of follow-up works~\cite{golovin2011adaptive,chen2013near,fujii2016budgeted,gabillon2013adaptive,gabillon2014large,gotovos2015nonmonotone,yu2016linear}.
Asadpour and Nazerzadeh~\cite{asadpour2016maximizing} considers the maximization of a monotone lattice-submodular function.
In this problem, each selected item has a stochastic state (a non-negative \emph{real} number). The target function accepts a vector of such numbers, and satisfies lattice-submodularity (defined in \Cref{sec:formulation}).
In their problem, only the states are stochastic, while the matroid constraint is on the set of selected items.
Fukunaga et al.~\cite{fukunaga2019stochastic} pushed one step further and allowed the constraints to be dependent on the state of items, but limited the set of states to be non-negative integers.


%% file: formulation.tex
\section{Preliminary}
\label{sec:formulation}

We start with some notations.
The description of the spot scheduling problem~\cite{yang2021scheduling} is delayed to \Cref{sec:spot_problem}, together with its modeling and reduction to the problem we consider.
In \Cref{sec:spot_reduction}, we explain how we reduce and manage to eliminate one critical assumption in the previous work by Fukunaga et al.~\cite{fukunaga2019stochastic}.

 Given two \(d\) dimensional vectors \(u, v\in [n]^{d}\), we write \(u \leq v\) if the inequality holds coordinate wise, i.e. \(\forall i\in [d], u(i) \leq v(i)\). Similarly, \(u\vee v\) and \(u \wedge v\) are defined coordinate wise: \((u\vee v)(i) = \max\{u(i), v(i)\}\), \((u \wedge v)(i) = \min\{u(i), v(i)\}\).
Consider a base set \([n]\), a matroid is defined to be an independent set \(\cI \subseteq 2^{n}\). This independent set needs to contain \(\emptyset\), and if \(A \in \cI\), so is every \(A' \subseteq A\). Furthermore, if \(A, B \in \cI\) and \(|A| > |B|\), then there exists an element \(x\in A\setminus B\) such that \(B \cup {x}\) is in \(\cI\). Particularly, for a \emph{partition matroid} \(\{\cI_{k}\}_{k\in [K]}\) where \(\cI_{i}\cap \cI_{j} = \emptyset, \forall i\neq j\), its independent set \(\cI\) is \(\{S|\forall k, S\cap \cI_{k} \leq 1\}\).

A function \(f: 2^{d} \to \mathbb{R}\) is \emph{submodular} if for every \(A, B\subseteq [d]\), \(f(A) + f(B) \geq f(A\cup B) + f(A\cap B)\). An equivalent definition is that for every \(A \subseteq B \subseteq [d]\) and \(e \in [d]\), \(f(A \cup \{e\}) - f(A) \geq f(B \cup \{e\}) - f(B)\). This definition is generalized to a domain of \([n]^{d}\), where function \(f: [n]^{d}\rightarrow \R^{ + }\) is called \emph{lattice-submodular} if \(f(u) + f(v) \geq f(u\wedge v) + f(u\vee v)\) holds for all \(u, v\in [n]^{d}\). Note that the \emph{lattice-submodularity} does not imply the property called \emph{DR-submodularity}, which is the diminishing marginal returns along the direction of \(\chi_{i}\) for each \(i \in I\), where \(\chi_{i}\in \{0, 1\}^{n}\), and only the \(i\)-th coordinate is 1. That is, \(f (u + \chi_{i}) - f (u) \geq f (v + \chi_{i}) - f (v)\) does not necessarily hold for all \(u, v \in [n]^{d}\) such that \(u \leq v\) and \(i \in [d]\) even if \(f\) is \emph{lattice-submodular}.
Function \(f: [n]^{d}\rightarrow \R^{ + }\) is called \emph{monotone} if \(f(u) \leq f(v)\) for all \(u \leq v\).

\subsection{Reduction and Eliminating an Assumption}
\label{sec:spot_reduction}
In order to eliminate the second assumption mentioned in \Cref{sec:spot_eliminate}, we introduce the notion of a ``size cap''. For each item \(i\) and a size cap \(b\), we define an item \((i, b)\), where
\(p_{(i, b)}(s) = \pi_{i}(s)\) when \(s < b\); \(p_{(i, b)}(s) = \sum_{s' \geq s}\pi_{i}(s')\) when \(s = b\); and \(0\) otherwise. The new reward function is exactly \(R_{(i, b)}(\cdot)\).

We will be using a time-indexed LP formulation following Ma~\cite{ma2018improvements}. Instead of making a decision at each time step, we do it at each remaining size level. When there is enough room, we take item \(i\) itself into consideration.
If the remaining size \(b\) is small, we are not able to get more reward for an item than when it has a size of \(b\). Therefore, instead of trying to include the original item \(i\), we include item \((i, b)\), which is item \(i\) with size cap \(b\).
Obviously, we can include each item at most once. To achieve this, we impose a partition matroid \(\{\mathcal{I}_{i}\}_{i\in [K]}\) on the items, where \(\mathcal{I}_{i} = \{(i, b)| \forall b\}\). For the remainder of this paper, we view each \((i, b)\) as an  item, and the conflict between them is captured by the partition matroid.


%% file: continuous.tex
\section{Continuous Optimization Phase}
\label{sec:continuous}
Like most submodular maximization problems, our algorithms consists of two phases, a continuous optimization phase and a rounding phase. In this section, we describe the former.

\subsection{Target Function}
 \label{sec:lp_exp}

 Given a lattice-submodular function \(f: [M]^{n} \rightarrow \mathbb{R}^{ + }\) and a distribution \(q_{S}\) of elements in set \(S \subseteq [n]\), we define a set-submodular function \(\bar{f}: 2^{n} \rightarrow \R_{ + }\), where \(\bar{f}(S) := \E_{r\sim q_{S}} [f(r)]\). This \(\bar{f}\) is guaranteed to be a monotone set-submodular function (See proof in \cite{asadpour2016maximizing}).
 Suppose the final selected (random) set is \(S\), the value we are interested in would be \(\mathbb{E}[\bar{f}(S)]\). Let \(\bar{x}\) be a vector, where \(\bar{x}(i)\) denotes the probability that item \(i\) is in \(S\).
 Using the well-established multi-linear extension, we define \(\bar{F}: 2^{n} \rightarrow \R^{ + }\), where \(\bar{F}(\bar{x}) = \sum_{S \subseteq [n]} \prod_{i\in S}\bar{x}_{i} \prod_{i' \notin S} (1 - \bar{x}_{i'}) \bar{f}(S)\). This is the target function we are maximizing. Evaluating the function \(\bar{F}\) can take exponential time, but it can be approximated within a multiplicative factor of \((1 + \epsilon)\) for any constant \(\epsilon > 0\). For simplicity, we assume \(\bar{F}(\bar{x})\) can be evaluated exactly in this paper, which is standard in the literature (e.g. see~\cite{calinescu2011maximizing}).

\subsection{Stochastic Knapsack Exponential Constraints}
 \label{sec:lp_exp}
 The exponential and polynomial constraints on \(\bar{x}\) are adapted from Ma~\cite{ma2018improvements}.
A group of exponential sized constraints describes the problem exactly. They are then relaxed to have a polynomial size, losing a factor of 2.
For ease of notation, we follow Ma~\cite{ma2018improvements} and view an stochastic item \(i\) as an equivalent Markovian bandit, a special one that can force us to keep pulling it for a certain period of time.
We use state \(u_{i}(k, s)\)  to indicate that arm \(i\) has been pulled \(k\) times, and the corresponding item has size \(s\).
From its initial state \(\rho_{i}\), a single pull would decide the size \(s\) of this job, and move to state \(u_{i}(1, s)\) respectively. We are then forced to keep pulling this arm (we will be using arm and item interchangeably) for the next \(s - 1\) steps, and the last of such pulls moves us to its termination state \(\emptyset_{i}\), and we can pull a new arm.
Denote the probability of moving from state \(u\) to state \(v\) with \(p_{u, v}\).
After the first pull of item \(i\), it moves to state \(u_{i}(1, s)\) (having size \(s\)) with probability \(p_{\rho_{i}, u_{i}(1, s)} = p_{i}(s)\).
 Therefore, if \(k < s\), a pull will transit it to state \(u_{i}(k+1, s)\) with probability \(p_{u_{i}(k, s), u_{i}(k+1, s)} = 1\). Otherwise, when \(k = s\), transit to state \(\emptyset_{i}\) with probability \(p_{u_{i}(k, s), \emptyset_{i}} = 1\), and we are allowed to pull a new arm.

Let \(\pi\) be a vector representing a joint state/node, where \(\pi_{i}\) denotes the state on item \(i\). Let \(S_{i} = \{u_{i}(*, *)\} \cup \{\rho_{i}, \emptyset_{i}\}\) for all \(i\in [n]\), the set of all states for arm \(i\), and \(\tilde{\fS} = S_{1} \times \cdots \times S_{n}\), the set of all possible (maybe invalid) joint states. Let \(\fS' = \{\pi \in \tilde{\fS}| \exists i\neq j, \pi_{i} \notin \{\rho_{i}, \emptyset_{i}\}, \pi_{j}\notin \{\rho_{j}, \emptyset_{j}\}\}\), the set of states where at least two arms are in the middle of processing at the same time, and \(\fS'' = \{\pi \in \tilde{\fS} |  \pi_{i} \neq \rho_{i} \text{ and  } \pi_{j} \neq \rho_{j}, i,j\in \cI_{k} \text{ for some } k\}\), the set of states where some conflicting arms (due to the partition matroid) have been started. Define \(\fS := \tilde{\fS} \setminus (\fS'\cup \fS'')\), which is the set of all valid states. Let \(I(\pi) = \{i | \pi_{i} \neq \emptyset_{i}\}\), the set of arms that could be played from state \(\pi\). Let \(\pi^{u}\) denote the joint node where the component corresponding to \(u\) is replaced by \(u\) (note \(u\) can correspond to only one component). Let \(y_{\pi, t}\) be the probability that we are at state \(\pi\) at time \(t\), and \(z_{\pi, i, t}\) the probability that we pulled arm \(i\) at time \(t\), when the current state was \(\pi\). Recall \(B\) is the total budget, we have the following basic constraints.
 \begin{align}
       \textstyle\sum_{i \in I(\pi)} z_{\pi, i, t} &\leq y_{\pi, t} && \pi\in \fS,  t \in [B]\label{eq:z}\\
        z_{\pi, i, t} &= y_{\pi, t}&& \pi \in \fS, i: \pi_{i} \in S_{i} \setminus\{\rho_{i}, \emptyset_{i}\}, t\in [B]\label{eq:zy}\\
       z_{\pi, i, t} & \geq 0 && \pi \in \fS, i\in [n], t \in [B]\label{eq:zz}
\end{align}

Let \(\fA_{i} = \{\pi \in \fS: \pi_{i} \notin \{\rho_{i}, \emptyset_{i}\}\}\), the joint node when arm \(i\) is in the middle of processing. Note \(\fA_{i}\) and \(\fA_{j}\) are disjoint for \(i\neq j\). We call arm \(i\) the \emph{active} arm. Let \(\fA = \bigcup_{i\in [n]} \fA_{i}\), the set of all states where some arm is \emph{active}. For a state \(\pi\in \fS\), let \(\fP(\pi)\) denote the subset of \(\fS\) that would transit to \(\pi\) with no play, which could happen when some arms turned inactive automatically: if \(\pi \notin \fA\), then \(\fP(\pi) = \{\pi\}\cup (\bigcup_{i\notin I(\pi)} \{\pi^{u} | u\in S_{i}\setminus\{\rho_{i}\}\}\); if \(\pi\in \fA\), then \(\fP(\pi) = \emptyset\). Suppose \(u\) corresponds to coordinate \(i\), define \(\textsf{Par}(u) = \{v \in S_{i}: p_{v, u} > 0\}\), the nodes that have a positive probability of transitioning to \(u\).  Then \(y\)-variables are updated as follows:
{\small
\begin{align}
     &\hspace{-1.7em}y_{(\rho_{1}, \dots, \rho_{n}), 0} = 1\label{eq:y0}\\
     y_{\pi, 0} &= 0, && \pi\in \fS \setminus \{(\rho_{1}, \cdots, \rho_{n})\}\label{eq:y1}\\
  y_{\pi, t} &= \hspace{-1em}\sum_{\pi' \in \fP(\pi)} \bigg(y_{\pi', t-1} - \hspace{-1em}\sum_{i \in I(\pi')}^{n} z_{\pi', i, t-1}\bigg)
  &&
     t > 0, \pi\in \fS \setminus \fA\label{eq:y_idle}\\
  y_{\pi, t} &= \hspace{-1.5em}\sum_{\rho_{i}\in \textsf{Par}(\pi_{i})} \hspace{-0.5em}\bigg(\sum_{\pi' \in \fP(\pi^{\rho_{i}})}^{n} \hspace{-1em}z_{\pi', i, t-1}\bigg)\cdot p_{\rho_{i}, \pi_{i}}, &&
     t > 0, i\in [n], \pi\in \fA_{i}, \pi_{i}\in \{u_{i}(1, * )\}\label{eq:y_d1}\\
  y_{\pi, t} &= \sum_{u\in \textsf{Par}(\pi_{i})}  z_{\pi^{u}, i, t-1}\cdot p_{u, \pi_{i}}, &&
     t > 0, i\in [n], \pi\in \fA_{i}, \pi_{i}\notin \{u_{i}(1, * )\}\label{eq:y_d2}
\end{align}}
\Cref{eq:y_idle} updates \(y_{\pi, t}\) for \(\pi\notin \fA\), i.e. joint nodes with no active arms. Such a joint node \(\pi\) can only come from a no-play from a joint node in \(\fP(\pi)\). \Cref{eq:y_d1,eq:y_d2} update \(y_{\pi, t}\) for \(\pi \in \fA\). To get to the joint node \(\pi\), we must have played arm \(i\) in previous step(s). In \Cref{eq:y_d1}, we consider the case if \(\pi_{i}\) is one of \(u_{i}(1, *)\). We were at \(\rho_{i}\) right before, so it is possible that in the last step, we switched to \(\pi^{\rho_{i}}\) from some joint node in \(\fP(\pi^{\rho_{i}})\) without playing an arm. In \Cref{eq:y_d2}, we consider other cases, in which arm \(i\) was played at time \(t - 1\). These equations guarantee that at each time step, \(y_{ * , t}\) form a distribution, i.e. \(\sum_{\pi \in \fS} y_{\pi, t} = 1\). Combining this with \Cref{eq:z}, we get \(\sum_{\pi\in S} \sum_{i\in I(\pi)} z_{\pi, i, t} \leq 1, \forall t\in [B]\).
\Cref{eq:z,eq:zy,eq:zz,eq:y0,eq:y1,eq:y_idle,eq:y_d1,eq:y_d2} form the exponential constraints. We also need to relate these constraints with \(\bar{x}\) (recall \(\bar{x}(i)\) is the probability that item \(i\) is included): \(\bar{x}(i) = \sum_{t}\sum_{u\in \mathcal{S}_{i}}\sum_{\pi\in \fS:\pi_{i} = u} z_{\pi, i, t}\), which is the last missing piece for our exponential program, denoted as \(\ExpP\).

\subsection{Stochastic Knapsack Polynomial Constraints}
\label{sec:lp_poly}
Obviously, we cannot solve this exponential program directly in polynomial time. In order to solve it, we relax the exponential program by disassemble the joint distribution of items. Let \(s_{u, t}\) be the probability that arm \(i\) is on node \(u\) at the beginning of time \(t\). Let \(x_{u, t}\) be the probability that we pull an arm on node \(u\) at time t. Suppose \(\cS = \bigcup_{i}S_{i}\), we have the following constraints between \(x_{u, t}\) and \(s_{u, t}\).
{\small
\begin{align}
     x_{u, t} &\leq s_{u, t} && u\in \cS, t\in [B]\\
     x_{u, t} &= s_{u, t} && u\in \textstyle\bigcup_{i\in [n]}\cS_{i}\setminus\{\rho_{i}, \emptyset_{i}\}, t\in [B]\\
     x_{u, t} &\geq 0 && u\in \cS, t\in [B]\\
     \textstyle\sum_{u\in \cS} x_{u, t} &\leq 1 && t\in [B]
\end{align}
}
We also need constraints~\labelcref{eq:partition} for the partition matroid of arms (recall \(\cI_{k}\) is a partition), and the state transition constraints~\labelcref{eq:s:u,eq:s:rho,eq:s:all}.
{\small\begin{align}
     \textstyle\sum_{i\in \cI_{k}} s_{\rho_{i}, 0} &\leq 1, \quad\forall \cI_{k} &&
     s_{\rho_{i}, 0} \geq 0,\quad  i\in [n]\label{eq:partition}\\
     s_{u, 0} &= 0 && u\in \cS\setminus \{\rho_{1}, \cdots, \rho_{n}\}\label{eq:s:u}\\
     s_{\rho_{i}, t} &= s_{\rho_{i}, t-1} - x_{\rho_{i}, t - 1} && t > 0, i\in [n]\label{eq:s:rho}\\
     s_{u, t} &= \textstyle\sum_{v\in \textsf{Par}(u)} x_{v, t-1}\cdot p_{v, u}&& t > 0, u\in \cS\setminus \{\rho_{1, \dots, \rho_{n}}\}\label{eq:s:all}
\end{align}
}
Relating these constraints with \(\bar{x}\): \(\bar{x}(i) = \sum_{t}\sum_{u\in \mathcal{S}_{i}} x_{u, t}\), we get the polynomial program \(\PolyP\). For any program \(P\in \{\PolyP, \ExpP\}\), let \(\OPT_{P}\)  denote its optimal value.

\subsection{Relating between the Exponential and the Polynomial Constraints}
\label{sec:relating_exp_poly}

This was given in Ma~\cite{ma2018improvements}, and we re-state for completes without proof. The direction from \(\ExpP\) to \(\PolyP\) is trivial.

\begin{theorem}
[reformation of Lemma 2.3 from Ma~\cite{ma2018improvements} ]
Given a feasible solution \(\{z_{\pi, i, t}\}, \{y_{\pi, t}\}\) to \(\ExpP\), we can construct a solution to \(\PolyP\) with the same objective value by setting \(x_{u, t} = \sum_{\pi\in \fS:\pi_{i} = u} z_{\pi, i, t}\), \(s_{u, t} = \sum_{\pi \in \fS: \pi_{i = u}}y_{\pi, t}\) for all \(i\in [n]\), \(u\in [0, 1]\), \(t\in B\). Thus, the feasible region of \(\PolyP\) is a projection of that of \(\ExpP\) onto a subspace and \(\OPT_{\ExpP} \leq \OPT_{\PolyP}\).
\end{theorem}

For the other direction, we construct a solution \(\{z_{\pi, i, t}, y_{\pi, t}\}\) of \(\ExpP\) from a solution \(\{x_{u, t}, s_{u, t}\}\) of \(\PolyP\), which obtains half its objective value. It will satisfy
\[\textstyle\sum_{\pi\in \fS: \pi_{i} = u} z_{\pi, i, t} = \frac{x_{u, t}}{2} \quad i\in [n], u\in \cS_{i}, t\in [B].\]
We define specific \(\{z_{\pi, i, t}, y_{\pi, t}\}\) over \(B\) iterations. On iteration \(t\):
\begin{itemize}
\item Compute \(y_{\pi, t}\) for all \(\pi  \in \fS\).
\item Define \(\tilde{y}_{\pi, t} = y_{\pi, t}\) if \(\pi \notin \fA\), and \(\tilde{y}_{\pi, t} = y_{\pi, t} - \sum_{a\in A} z_{\pi, i, t}\) if \(\pi \in \fA_{i}\) for some \(i\in [n]\) (if \(\pi \in \fA_{i}\), then \(\{z_{\pi, i, t} : a\in A\}\) is already set in a previous iteration).
\item For all \(i\in [n]\), define \(f_{i, t} = \sum_{\pi \in \fS: \pi_{i} = \rho_{i}} \tilde{y}_{\pi, t}\).
\item For all \(i\in [n]\), \(\pi \in\fS\) such that \(\pi_{i} = \rho_{i}\), and \(a\in A\), set \(z^{a}_{\pi, i, t} = \tilde{y}_{\pi, t}\cdot \frac{1}{2} \cdot \frac{x_{\rho_{i, t}}}{f_{i, t}}\).
\item For all \(i\in [n]\), \(\pi\in \fS\) such that \(\pi_{i} = \rho_{i}\) and \(\pi_{j} \in \{\rho_{j}, \phi_{j}\}\) for \(j \neq i\), define \(g_{\pi, i, t} = \sum_{\pi' \in \cP(\pi)} z_{\pi', i, t}\).
\item For all \(i\in [n]\), \(u\in \cS_{i}\setminus \{\rho_{i}\}\), \(\pi\in \fS\) such that \(\pi_{i} = u\), and \(a\in A\), set \(z^{a}_{\pi, i, t+\textsf{depth}(u)} = g_{\pi^{\rho_{i}}, i, t} \cdot (x^{a}_{u, t + \textsf{depth}(u))})/x_{\rho_{i}, t}\).
\end{itemize}

\subsection{Solve the Continuous Optimization Problem}
\label{sec:solve_continuous}
   In order to solve \(\PolyP\), we follow Fukunaga et al.~\cite{fukunaga2019stochastic} and use the Stochastic Continuous Greedy algorithm. This algorithm maximizes the multi-linear extension \(G\) of a monotone set-submodular function \(g\) over a solvable downward-closed polytope. A polytope \(\PP \subseteq [0, 1]^{\NN}\) is considered \emph{solvable} if we can find an algorithm to optimize linear functions over it, and downward-closed if \(x \in \PP\) and \(0 \leq y \leq x\) imply \(y \in \PP\).
In our case, \(\PP\) is solvable due to its linearity, and that solving a linear program falls in polynomial time. Note \(\PP\) is down-monotone. The algorithm involves a controlling parameter called \emph{stopping time}. For a stopping time \(0 < b \leq 1\), the algorithm outputs a solution \(x\) such that \(x/b \in \PP\), while \(G(x) \geq (1 - e^{-b} - O(n^{3}\delta))\max_{y\in Q} G(y)\), where \(n\) is the size of the set over which \(g\) is defined and \(\delta\) is the step size used in the algorithm. Here \(\PP\) is assumed to include the characteristic vector of every singleton set.
\begin{theorem}[reformation of Theorem 3 from Fukunaga et al.~\cite{fukunaga2019stochastic}]
  \label{thm:fukunaga}
If the stochastic continuous greedy algorithm with stopping time \(b = 1/2\in (0, 1]\) and step size \(\delta = o(|I|^{-3})\) is applied to program \(\PolyP\), then the algorithm outputs a solution \(x\in b\PP\) such that \(\bar{F}(\bar{x}) \geq (1 - e^{-b} - o(1)) f_{\textsf{avg}} (\pi^{ * })\simeq 0.3935 f_{\textsf{avg}} (\pi^{ * })\) for any adaptive policy \(\pi^{ * }\).
\end{theorem}

%% file: rounding.tex
\section{Rounding Phase}
\label{sec:rounding_phase}
Now that we have a fractional solution \(x\), we proceed to round it to an integral policy (notice the fractional solution has already been scaled by a factor of 2). We need a variant of the contention resolution scheme that was introduced as a general framework for designing rounding algorithms that maximizes expected submodular functions (\cite{chekuri2014submodular,feldman2013maximization,feldman2011unified}). The variant is an extension from a set submodular function to a lattice-submodular function, first introduced in Fukunaga et al.~\cite{fukunaga2019stochastic}. We include its definition here for self-containment.

\subsection{Contention Resolution Scheme}
\label{sec:CRS}
A contention resolution scheme (CRS) accepts a \emph{pairwise independently} rounded solution which may violate some constraints, and fixes it without losing too much on expectation. Let \(f: [B]^{n}\rightarrow \R_{+}\) be a monotone lattice-submodular function and the probability distribution \(q_{i}: [B] \rightarrow [0, 1]\) on \([B]\) be given for each \(i\in \{1, \dots, n\}\}\). We write \(v\sim q\) if \(v\in [B]^{n}\) is a random vector such that, for each \(i\in \{1, \dots, n\}\), the corresponding component \(v(i)\) is determined independently as \(j\in [B]\) with probability \(q_{i}(j)\). This is the independently rounded solution we feed into a CRS. Let \(\cF \subseteq [B]^{n}\) be a downward-closed subset of \([B]^{n}\), and let \(\alpha \in [0, 1]\). We have the following definition for a \(\alpha\)-CRS, its monotonicity, and one key property.
\begin{definition}
[\(\alpha\)-Contention Resolution Scheme (\(\alpha\)-CRS)]
A mapping \(\psi: [B]^{n} \rightarrow \cF\) is an \(\alpha\)-CRS with respect to \(q\) if it satisfies:
\begin{enumerate}
\item \(\psi(v)(i)\in \{v(i), 0\}\) for each \(i\in [n]\);
\item if \(v\sim q\), then \(\Pr[\psi(v)(i) = j|v(i) = j] \geq \alpha\) holds for all \(i\in I\) and \(j\in B\). The probability is based on randomness both in \(v\) and in \(\psi\) when \(\psi\) is randomized.
\end{enumerate}
\end{definition}

\begin{definition}
[monotone \(\alpha\)-CRS]
An \(\alpha\)-CRS \(\psi\) is considered \emph{monotone}, if, for each \(u, v\in [B]^{n}\) such that \(u(i) = v(i)\) and \(u \leq v\), \(\Pr[\psi(u)(i) = u(i)] \geq \Pr[\psi(v)(i) = v(i)]\) holds. The probability is based only on the randomness of \(\psi\).
\end{definition}

\begin{lemma}
[Theorem 4 from Fukunaga et al.~\cite{fukunaga2019stochastic}]
If \(\psi\) is a monotone \(\alpha\)-CRS with respect to \(q\), then \(\E_{v\sim q}[f(\psi(v))] \geq \alpha \E_{v\in q}[f(v)]\).
\label{lemma:CRS}
\end{lemma}

\subsection{Rounding Algorithm}
\label{sec:rounding_alg}

To fit in the contention resolution scheme, we need to first round everything independently. This means for each pair \((i, t)\), item \(i\) is scheduled at time \(t\) with probability \(x_{\rho_{i}, t}\).
Now we have a set \(R' = \{(i, t)\}\) of proposed item time pairs. We sort the set according to \(t\), and include the items one by one. Intuitively, for a pair \((i, t)\), we will only include item \(i\) if time \(t\) is available and does not invalidate the solution, i.e. each item is scheduled at most once, and at most one item from each partition. After including it in our solution, we get its realized size, and mark the corresponding time slots unavailable.

The main problem of this naive approach is that it does not exhibit monotonicity, which is a subtle but critical requirement for a CRS. To fix it, we schedule \emph{phantom item} \(i\) even when we cannot fit it. We \emph{simulate} its inclusion, and sample its size \(\size_{i}\) should it be included. We also mark those time slots corresponding to this \emph{phantom item} unavailable, even when they are actually unoccupied. This seemingly wasteful step ensures that the rounding scheme is monotone. The final rounding algorithm is described in \Cref{alg:rounding}.

\begin{algorithm}[htbp]
  \ForEach{pair \((i, t)\)}{
    Sample \((i, t)\) with probability \(x_{\rho_{i}, t}\), and gets \(\emptyset\) otherwise\;
    \lIf{not get \(\emptyset\)}{
      \(I \gets I \cup \{(i, t)\}\)
    }
  }
  Sort \(I\) according to a non-decreasing ordering of \(t\), break ties uniformly at random\;
  \(C = 0, S = \emptyset\), mark all times slots available\;
  \For{\((i, t) \in I\)}{
    \eIf{time slot \(t\) is available \textbf{and} item \(i\) does not violate constraints} {
      Include item \(i\) and observe \(s_{i}\)\;
    }{
      Simulate including item \(i\), and observe \(s_{i}\)\;
    }
    Mark time slots from \(t\) to \(t + s_{i}\) unavailable\;
  }
     \caption{Rounding Algorithm}
     \label{alg:rounding}
\end{algorithm}

The remaining of this section is devoted to proving the following theorem, which combined with \Cref{thm:fukunaga} leads to our main result.
\begin{theorem}
Let \(\pi\) denote \Cref{alg:rounding}, and \(x\) denote the solution we get from \(\PolyP\). Then \(f_{\text{avg}}(\pi)\geq \bar{F}(\bar{x})/2\).
\label{theorem:half}
\end{theorem}

To prove \Cref{theorem:half}, we define two mappings \(\sigma(\cdot)\) and \(\omega(\cdot)\), where the first corresponds (roughly) to the step that maps \(x\) to \(I\) in \Cref{alg:rounding}, and \(\omega(\cdot)\) corresponds to the mapping (CRS) from set \(I\) to the final output. The mapping \(\sigma(\bar{x})\) receives a real vector \(\bar{x}\in [0, 1]^{n}\) and returns a random vector \(v\in [B]^{n}\). From each partition \(\cI_{k}\), we pick at most one \(i\), each \(i\in \cI_{k}\) is picked with probability \(\bar{x}(i)\). If it is picked, the \(i\)-th component \(v(i)\) independently takes value \(j\) with probability \(p_{i}(j)\), and 0 otherwise, which happens with probability \(1 - \sum_{j}p_{i}(j)\). This captures the construction of set \(I\) (only the item part, note \(\Pr[\sigma(x)(i) > 0] = \Pr[\exists t, \textrm{s.t.} (i, t)\in I]\)), together with the random outcome of the item. The mapping \(\omega(\cdot)\) maps \(v\in [B]^{n}\) to \(w\in [B]^{n}\). To mimic \Cref{alg:rounding}, we first assign time value \(t(i)\) to each component \(v(i)\), according to \(x_{\rho_{i}, t}\). Based on \(t(i)\), we form a precedence ordering \(\prec\) between \(i\) after random tie breaking (a random tie breaking is crucial). Then, we set \(\omega(v)(i) =0\) if there exists a component \(j \prec i\) such that \(t(j) \leq t(i) < t(j) + v(i)\), and \(w(v)(i) = v(i)\) otherwise. We can observe that given input \(x\), \Cref{alg:rounding} outputs exactly \(\omega(\sigma(x))\) if the random realized sizes of items are the same.
In order to prove \Cref{theorem:half}, we need the following two lemmas. The first, whose proof in Fukunaga~\cite{fukunaga2019stochastic}, corresponding to the independent rounding step, and the second corresponding to the CRS step.

\begin{restatable}{lemma}{lemmaSigma}
\(\E[f(\sigma(x))] \geq \bar{F}(\bar{x})\) holds for any \(x\in P\).
\label{lemma:sigma}
\end{restatable}
\begin{restatable}{lemma}{lemmaCRS}
\(\omega\) is a \(1/2\)-CRS with respect to \(\bar{x}\).
\label{lemma:CRS-omega}
\end{restatable}
\begin{restatable}{lemma}{lemmaMonoCRS}
The \(1/2\)-CRS \(\omega\) is monotone.
\label{lemma:mono-CRS-omega}
\end{restatable}

\Cref{lemma:sigma} is trivially true by the definition of \(\bar{F}(\bar{x})\), which is the common starting point of contention resolution scheme. We first prove \(\omega\) is a \(1/2\)-CRS.

\begin{proof}[Proof of \Cref{lemma:CRS-omega}]
  Recall there are two properties needed for an \(\alpha\)-CRS. The first property is obviously correct due to the definition of \(\omega(\cdot)\). The second property needs to prove \(\Pr[\omega(v)(i) = j|v(i) = j] \geq 1/2\). In the language of the rounding algorithm, let \(\textsf{Drop}_{i, t}\) denotes the event (respect to the randomness in \(\omega\) and \(v\)) that we drop the pair \((i, t)\). It is the same as proving
\[\Pr[\textsf{Drop}_{i, t}|\text{item }i\text{ is selected at time  }t] \leq \frac{1}{2}.\]

Due to the way we round the solution, item \(i\) may be included more than once (at different times), and more than one item from the same partition may be included. Consider an item \(j\) at time \(t'\) (maybe the same as \(i\)) that could affect the pruning of item \(i\) at time \(t\).
Define \((j, t') \prec (i, t)\) if \(t' < t\), or \(t' = t\) and \(j \prec i\).
It is clear that \((j, t')\) will affect \((i, t)\) if and only if \((j, t')\prec (i, t)\)
We slightly abuse notation, and let \(\textsf{Drop}_{i, t}(j)\) denote the probability that the item \(j\) \emph{can} causes the drop out of item \(i\) if a copy of it is scheduled at time \(t\). Note this does not depend on whether item \(i\) is scheduled on \(t\) or not. We have:

\begin{lemma}
\[\textsf{Drop}_{i, t}(j) \leq \frac{1}{2}\sum_{u\in \{\emptyset_{j}\} \cup\{u_{j}(* , * )\}} x_{u, t} + \frac{1}{2} x_{\rho_{j}, t}.\]
\label{lemma:case}
\end{lemma}
\begin{proof}[Proof of \Cref{lemma:case}]
There are two cases and we bound the probability of dropping in each case.
\begin{enumerate}[{Case} 1.]
  \item \(j\) belongs to the same partition as \(i\),
  \item \(j\) belongs to a different partition.
\end{enumerate}

For the first case, the probability that it makes \((i, t)\) invalid is
\begin{align*}
  \textsf{Drop}_{i, t}(j)
  \leq&\frac{1}{2}(s_{\rho_{j}, 0} - s_{\rho_{j}, t}) + \Pr[\text{item } j \text{ is considered before }i]\cdot \frac{1}{2} x_{\rho_{j}, t}\\
       \leq&  \frac{1}{2}\sum_{u\in \{\emptyset_{j}\} \cup\{u_{j}(* , * )\}} x_{u, t} + \frac{1}{2} x_{\rho_{j}, t}.
\end{align*}
The first term is the probability that at least one item \(j\) is scheduled before time \(t\). Note this is actually an union bound due to our independent rounding. The second term is the probability that it is scheduled at time \(t\), but will invalidate \(i\) since \(j \prec i\). The second equality comes from the fact that if item \(j\) is scheduled some time before \(t\), then it must be at some state at time \(t\) that is not the starting state \(\rho_{j}\). In other words, either the end state \(\emptyset_{j}\) or some transient state \(u_{j}(*, *)\).

For the second case,
fix \(j\), it can only drop \(i\) if it marked time slot \(t\) unavailable. The probability is
     \begin{align*}
       \textsf{Drop}_{i, t}(j)
       \leq&\frac{1}{2}\sum_{t' = 1}^{t- 1}  x_{\rho_{j}, t'}\cdot \Pr[\size_{j} \geq t - t']
       + \Pr[\text{item } j \text{ is considered before }i]\cdot \frac{1}{2} x_{\rho_{j}, t}\\
\leq & \frac{1}{2}\sum_{t'=1}^{t - 1} x_{\rho_{j}, t'}\cdot \Pr[\size_{j} \geq t - t'] + \frac{1}{2}x_{\rho_{j}, t}.
     \end{align*}
     The first term is a summation of all the possible starting point of job \(j\), times the probability that it will mark time slot \(t\) unavailable. Note this is also a union bound since there can be more than one copy of item \(j\) due to independent rounding. The second term is the probability that item \(j\) is also scheduled at time \(t\), but is considered before \(i\), i.e. \(j \prec i\), which marks time slot \(t\) unavailable for \(i\). We focus on the first term,
\begin{align*}
   &\sum_{t' = 1}^{t - 1} x_{\rho_{j}, t'} \cdot \Pr[\size_{j} \geq t - t']
   =\sum_{t' = 1}^{t - 1} \sum_{\tau = t - t'}^{B - t} x_{\rho_{j}, t'}\Pr[\size_{j} = \tau]
   =\sum_{t' = 1}^{t - 1} \sum_{\tau = t - t'}^{B - t} x_{u_{j}(1, \tau), t' + 1}\\
   =&\sum_{t' = 1}^{t - 1} \sum_{\tau = t - t'}^{B - t} x_{u_{j}(t - t', \tau), t}
   \leq\sum_{u\in \{\emptyset_{j}\} \cup\{u_{j}(* , * )\}} x_{u, t}.
\end{align*}
The last inequality holds because the index set of the summation on the left is a subset of that on the right.
\end{proof}

Therefore, the total probability that item \(i\) is blocked by any item is upper bounded by the union bound:
     \begin{align*}
       &\textsf{Drop}_{i, t}
       = \sum_{j)} \textsf{Drop}_{i, t}(j)
\leq \frac{1}{2}\sum_{j}\sum_{u\in \{\emptyset_{j}\}\cup \{u_{i}( * , * ) \}}x_{u, t} + \frac{1}{2}\sum_{j\in [n]} x_{\rho_{j}, t}\\
\leq & \frac{1}{2}\sum_{j\in [n]}\sum_{u\in \{\emptyset_{j}\}\cup \{u_{i}( * , * ) \}}x_{u, t} + \frac{1}{2}\sum_{j\in [n]} x_{\rho_{j}, t}
\leq  \frac{1}{2}(1 - \sum_{j\in n} x_{\rho_{j}, t}) + \frac{1}{2}\sum_{j\in [n]} x_{\rho_{j}, t}
=  \frac{1}{2}.
     \end{align*}%
   \end{proof}
Lastly, we show \(\omega\) is monotone in \Cref{sec:lemma-proof:mono-CRS-omega}.
With everything ready, we can now prove \Cref{theorem:half}, which combined with \Cref{thm:fukunaga} leads to the main claim.

\begin{proof}[Proof of \Cref{theorem:half}]
The output \(r\) of \Cref{alg:rounding} satisfies \(\E[f(r)] = \E[f(\omega(\sigma(x)))]\), and its feasibility is guaranteed by the algorithm. By \Cref{lemma:CRS-omega} and \Cref{lemma:mono-CRS-omega}, \(\omega\) is a monotone \(1/2\)-CRS with respect to \(q\), where \(q\) is the probability defined in \Cref{lemma:CRS-omega}. Moreover, \(\sigma(x) \sim q\) holds. By \Cref{lemma:CRS}, \(\E[f(\omega(\sigma(x)))] \geq \E[f(\sigma(x))]/2\). Using \Cref{lemma:sigma}, we get \(f_{\text{avg}}(\pi) = \E[f(r)]= \E[f(\omega(\sigma(x)))] \geq \E[f(\sigma(x))]/2\geq \bar{F}(\bar(x))/2\).
\end{proof}


%% file: conclusion.tex
\section{Conclusion}
\label{sec:dissertation-conclusion}
We consider the well studied correlated stochastic knapsack problem, generalizing its target function with submodularity to capture diminishing returns. An extra partition matroid constraint is added to generalize it and resolve an open question raised in a previous work to eliminate an assumption. We also make improvement on the approximation ratio. There is still a gap of \(2\) comparing to the variant with linear target function and we leave it as an open problem to close the gap.


%% file: proofs.tex
\section{Missing Proofs}
\label{sec:missing-proofs}



\subsection{Proof of \Cref{lemma:mono-CRS-omega}}
\label{sec:lemma-proof:mono-CRS-omega}

\lemmaMonoCRS*

   \begin{proof}[Proof of \Cref{lemma:mono-CRS-omega}]
     Suppose vectors \(u, v\in [B]^{n}\) satisfies \(u \leq v\), and \(u(i) = v(i) = j > 0\). We only need to show \(\Pr[\omega(u)(i) = j ]\geq \Pr[\omega(v)(i) = j]\), where randomness is with respect to the choice of time and ordering. Let \(I\) denote the partition that includes item \(i\). In this case,
\[\Pr[\omega(u)(i) = j] = \sum_{t = 1}^{B} x_{\rho_{i}, t}\prod_{\tau = 1}^{t-1}(1 - x_{\rho_{i}, \tau}) \left(\prod_{i'\notin I} \prod_{t' = t - u(i')}^{t}(1 - \Pr[i' \prec i | t(i) = t] x_{\rho_{i'}, t'})\right) \cdot \left(\prod_{i'\in I}\prod_{t'=0}^{t-1}(1 - x_{\rho_{i'},t'})\right). \]
This is a summation over all possible time slot that the first copy of item \(i\) is scheduled.
For simplicity, we define the following:
\begin{align*}
  a_{i, t} & = \sum_{t = 1}^{B} x_{\rho_{i}, t}\prod_{\tau = 1}^{t-1}(1 - x_{\rho_{i}, \tau}) \\
  b_{i, t, u} & = \prod_{i'\notin I} \prod_{t' = t - u(i')}^{t}(1 - \Pr[i' \prec i | t(i) = t] x_{\rho_{i'}, t'})\\
  c_{i, t} & = \prod_{i'\in I}\prod_{t'=0}^{t-1}(1 - x_{\rho_{i'},t'})
\end{align*}
which re-writes \(\Pr[\omega(u)(i) = j]\) as \(\sum_{t=1}^{B}a_{i, t}b_{i, t, u}c_{i, t}\).

The part in the first large bracket (\(b_{i, t, u}\)) is the probability that non of the items in a different partition prunes item \(i\) at time \(t\). The part in the second large bracket (\(c_{i, t}\)) is that for items in the same partition. Such multiplication of probability is only possible due to the phantom items and the independence they brought.

We wish to prove that \(\Pr[\omega(u)(i) = j] \geq  \Pr[\omega(v)(i)=j]\), and we instead prove that coordinately \(a_{i, t}b_{i, t, u}c_{i, t} \geq a_{i, t}b_{i, t, v}c_{t}\geq 0\) always holds. In fact
\begin{align*}
  &\frac{a_{i, t}b_{i, t, u}c_{i, t}}{a_{i, t}b_{i, t, v}c_{i, t}}\\
  =& \frac{b_{i, t, u}}{b_{i, t, v}}\\
  =& \frac{\prod_{i'\notin I} \prod_{t' = t - u(i')}^{t}(1 - \Pr[i' \prec i | t(i) = t] x_{\rho_{i'}, t'})}{\prod_{i'\notin I} \prod_{t' = t - v(i')}^{t}(1 - \Pr[i' \prec i | t(i) = t] x_{\rho_{i'}, t'})}\\
  =& \prod_{i'\notin I}\frac{ \prod_{t' = t - u(i')}^{t}(1 - \Pr[i' \prec i | t(i) = t] x_{\rho_{i'}, t'})}{\prod_{t' = t - v(i')}^{t}(1 - \Pr[i' \prec i | t(i) = t] x_{\rho_{i'}, t'})}\\
  =& \prod_{i'\notin I}\left( \prod_{t' = t - v(i')}^{t - u(i')}\frac{1}{1 - \Pr[i' \prec i | t(i) = t] x_{\rho_{i'}, t'}}\right)\\
  \geq & 1
\end{align*}

So \(a_{i, t}b_{i, t, u}c_{i, t} \geq a_{i, t}b_{i, t, v}c_{t}\). We sum both sides over \(t\), which leads to \(\Pr[\omega(u)(i) = j ]\geq \Pr[\omega(v)(i) = j]\).
\end{proof}


%% file: scheduling.tex
\section{Scheduling ML Jobs on Cloud Spot Instances}
\label{sec:spot_problem}
\textbf{Cloud Computing Instance Characteristics:}
Demands for cloud resources display large fluctuations across time and availability zones~\cite{calzarossa2016workloads,saripalli2011load}.
During times of low actual demand, cloud vendors make unused resources available to user as cheaper entities, a.k.a. spot instances\footnote{Known as \emph{spot instances} by Amazon Web Services (AWS), \emph{low-priority VMs} by Microsoft Azure, \emph{preemptible instances} by Google Cloud and \emph{transient} virtual machines/servers in some literature. We refer to all such revocable computing instances as \textit{spot instances}.}, that may be interrupted, so they can take them back when demands surge.
In practice, spot instances are often available at up to 70\%-90\% discounts compared to their on-demand equivalences~\cite{harlap2017proteus,amazon-spot-price}, making them an economical option if interruptions can be handled.

\textbf{Machine Learning Characteristics:}
An ML training algorithm is usually an iterative algorithm (each iteration is also known as an epoch) and it produces a better estimate of the model parameters with each iteration, usually with diminishing marginal returns.
If a long-running ML training job is interrupted prematurely, model parameter estimates from the latest successfully completed iteration are still a valid model instance, hence interruptions can be tolerated with adequate planning. We try to model and answer the following question from a theoretical perspective:
\begin{quote}
How can ML training jobs be scheduled and executed on interruptible but relatively inexpensive spot instances to increase their cost efficiency?
\end{quote}

Now we give a rigid definition of the spot scheduling problem. For justification of the modeling, please refer to Yang et al.~\cite{yang2021scheduling}. $\mathcal{N}$ jobs need to be scheduled on $\mathcal{M}$ instances, where the instances may have different CPU/RAM configurations, i.e.\ have different speeds for various jobs, or come from different available zones, i.e.\ have different interruption patterns. Each instance has a finite supply, and without loss of generality, we assume different copies to be separate instances.
For a spot instance \(i\), a job can run on it for a period time before interruption, which follows a given distribution independent of each other.
Let $\pi_{i, s}$ be the probability that instance \(i\) costs exactly $s$ dollars before it gets interrupted.
When we schedule a job \(j\) on an instance \(i\), we can also specify a budget cap.
Let $R_{(j, i)}(s)$ denotes the progress of job \(j\) achieves when \(s\) dollars have been spent, before the last check point, e.g.\ the number of trained epochs.
Notice the function \(R_{(j, i)}(\cdot)\) is monotone, i.e.\ $R_{(j, i)}(s) \leq R_{(j, i)}(s')$ if \(s \leq s'\).
In practice, when we schedule job \(j\) onto instance \(i\), it cannot start training immediately. Some processing time is wasted on environment setup and checkpoint restoration, which does not count towards progress.
This is captured by setting \(R_{(j, i)}(s) = 0\) if \(s\) dollars is not enough to finish the first epoch.
When a job gets interrupted, we can reschedule it on a different instance, starting from the latest checkpoint.
The total utility model this as a submodular function.
With a given budget $B$, we would like to maximize the total expected utility of all jobs.

\subsection{Reduction}
\label{sec:scheduling_reduction}

The reduction from the scheduling problem to the final knapsack problem is as follows. For each job \(j\), instance \(i\) and budget cap \(b\), we define an item \((j, i, b)\), where
\(p_{(j, i, b)}(s) = \pi_{i, s}\) when \(s < b\); \(p_{(j, i, b)}(s) = \sum_{s' \geq s}\pi_{i, s'}\) when \(s = b\); and \(0\) otherwise. The new reward function is exactly \(R_{(j, i, b)}(\cdot)\). Notice a job may be scheduled on multiple instances sequentially due to interruptions, but for each instance, only a single job can be scheduled on it, and a specific budget cap can be chosen, we further impose a partition matroid \(\{\mathcal{I}_{i}\}_{i\in [K]}\) on the items, where \(\mathcal{I}_{i} = \{(j, i, b)| \forall j, \forall b\}\).
